%% file: strictp.tex
\documentclass[copyright,creativecommons]{eptcs}
\providecommand{\event}{E. Formenti (Ed.): AUTOMATA \& JAC 2012} 

\usepackage{amsfonts}
\usepackage{mathrsfs}
\usepackage{graphicx}
\usepackage{xspace}

\usepackage{comment}
\usepackage{amssymb,amscd}
\usepackage{stmaryrd}
\usepackage{amsmath}
\usepackage{pifont}
\usepackage{tabls,subfig}
\usepackage{tikz}
\usepackage{enumerate}

\usepackage{color}

\input{env.tex}

\newcommand{\ie}{i.e.\@\xspace}
\newcommand{\wrt}{w.r.t.\@\xspace}

\newcommand{\modulo}[2]{\left[ #1\right]_{#2}}
\newcommand{\zm}{\mathbf{Z}_m}

\newcommand{\zpk}{\ensuremath{\mathbf{Z}_{p^k}}\xspace}

\newcommand{\zpini}{\mathbf{Z}_{p_i^{n_i}}}

\title{Strictly Temporally Periodic Points in Cellular Automata\thanks{This work has been supported by the PRIN/MIUR project ``Formal Languages and 
Automata: Mathematical and Applicative Aspects''}
}
\author{Alberto Dennunzio
\institute{Universit\`a degli Studi di Milano--Bicocca\\
Dipartimento di Informatica, Sistemistica e Comunicazione,\\
Viale Sarca 336, 20126 Milano (Italy)}
\email{dennunzio@disco.unimib.it}
\and
Pietro Di Lena \qquad\qquad Luciano Margara
\institute{Universit\`a degli Studi di Bologna,
Dipartimento di Scienze dell'Informazione,\\
via Mura Anteo Zamboni 7, 40127 Bologna (Italy)}
\email{\quad dilena@cs.unibo.it \quad\qquad margara@cs.unibo.it}
}
\def\titlerunning{Strictly Temporally Periodic Points in Cellular Automata}
\def\authorrunning{A. Dennunzio, P. Di Lena \& L. Margara}
\begin{document}
\maketitle
\begin{abstract}
We study the set of strictly temporally periodic points in surjective cellular automata, \ie,  the set of those configurations which are temporally periodic for a given automaton but are not spatially periodic. This set turns out to be residual for equicontinuous surjective cellular automata, dense for almost equicontinuous surjective cellular automata, while it is empty for the positively expansive ones. In the class of additive cellular automata, the set of strictly temporally periodic points can be either dense or empty. The latter happens if and only if the cellular automaton  is topologically transitive.
\end{abstract}

\section{Introduction}
Cellular Automata (CA) are a simple formal model for complex systems, \ie, those systems defined by a multitude of simple objects which cooperate to build a (unexpected) complex global behavior by é means of local interactions. CA are used in many scientific fields ranging from biology to chemistry or from physics to computer science (see for instance~\cite{chemical,tumour,chop,chaudhuri97,Wolfram02,FD08}). 

A cellular automaton is made of an infinite set of finite automata distributed over a regular 
lattice (usually $\z^n$, with $n=1$ in this work). All automata are identical. Each automaton assumes a state, 
chosen from a finite set, called the set of states or the alphabet. A configuration is a snapshot of all the states of the automata. A local rule updates the state of an automaton on the basis of its current state and those of a fixed set of neighboring 
automata. All the automata of the lattice are updated synchronously and this global updating gives rise to a discrete dynamical system on the configuration space. 

Several dynamical properties of CA have been studied during the last two decades (see for instance~\cite{capka09, DLT, Dennunzio10cie, noi, IC12} for recent results and an up-to-date bibliography). On the basis of the well-known results from~\cite{Ku97} and~\cite{blanchard97}, one-dimensional CA can be classified from the most stable to the most unstable behavior (dynamical complexity classification):
\begin{itemize}
\item equicontinuous CA;
\item non equicontinuous CA admitting an equicontinuous configuration (pure almost equicontinuous CA);
\item sensitive to the initial conditions but non topologically mixing CA;
\item topologically transitive but non positively expansive CA;
\item positively expansive CA.
\end{itemize}

Another significant information about the dynamical behavior of CA (and of general discrete dynamical systems) is given by temporally periodic configurations. If the set of the temporally periodic configurations of a cellular automaton is dense, then the cellular automaton has \emph{dense periodic points (DPO)}. Together with topological transitivity and sensitivity to the initial conditions, DPO is also a fundamental property of the popular Devaney's definition of chaos for discrete dynamical systems~\cite{De89}. In the CA setting, DPO is shared by both the two classes of the surjective almost equicontinuous and closing automata. One of the most challenging, long-standing open problem in CA concerns DPO: is it enjoyed by all surjective CA \cite{boyle99,BT00,BL07,Boyle08}? 
If the answer is affirmative, then chaotic behavior in CA reduces to transitivity, due to the fact that transitive CA are both sensitive and surjective.

We can classify two distinct types of of temporally periodic configurations in CA: the temporally periodic configurations that are also spatially periodic (\emph{jointly periodic points}) and the ones that are not (\emph{striclty temporally periodic points}). In this paper we deal with the set of strictly temporally periodic points in CA. Among all temporally periodic configurations, the strictly periodic configurations are the ones that provide more information about the CA dynamical behavior. 
In fact, if surjective CA have DPO, the set of jointly periodic configurations  is dense for surjective CA in any class of dynamical complexity~\cite{ADF09}. This does not happen for the set of strictly temporally periodic configurations. Indeed, for surjective CA belonging to a certain class, the size of this set turns out to be inversely related to the dynamical complexity of that class. More precisely, in this paper we show that 
\begin{itemize}
\item surjective equicontinuous CA exhibit a residual set of strictly temporally periodic configurations (Proposition~\ref{equic});
\item for almost but non equicontinuous surjective CA  the set of strictly temporally periodic configurations  is dense (Proposition~\ref{almost});
\item positively expansive CA admit no strictly temporally periodic configuration (Proposition~\ref{exp}).
\end{itemize}
We also study the set of strictly temporally periodic configurations for the class of \emph{additive CA}, \ie, those CA whose local rule is defined by an additive function. Despite their simplicity, which  
makes it possible a detailed algebraic analysis, additive CA exhibit many of the complex features of general CA. In this settings, we prove that (Proposition~\ref{spadd})
\begin{itemize}
\item the set of strictly temporally periodic configurations can be either empty or dense;
\item the set of strictly temporally periodic configurations is empty if and only if the additive CA is topologically transitive (or, equivalently, topologically mixing).
\end{itemize}
The paper is organized as follows. In Section \ref{background} we introduce the basic notation and the general background on Cellular Automata. In Section \ref{generalCA} and Section \ref{additiveCA} we prove our main results for general CA and additive CA, respectively. Section \ref{conclusions} is devoted to the final remarks.

\section{Basic Notions}\label{background}
In this section, we briefly recall standard definitions about CA
as discrete dynamical systems. We begin by introducing some general notation we will use throughout the rest of the paper.
\\\\
For all $i,j\in\z$ with $i\leq j$ (resp., $i<j$), we use the notation $[i,j]=\{i,i+1,\ldots,j\}$ (resp., $[i,j)=\{i,i+1,\ldots,j-1\}$) to denote the interval of integers between $i$ and $j$. Let us define $\n_+$ as the set of positive integers.  For any pair of integers $n,m>0$, by $n\mid m$ and $n\nmid m$ we denote the fact that $n$ divides $m$ and $n$ does not divide $m$, respectively.

Let $A$ be a finite alphabet with at least two elements. A \emph{configuration} is a function from $\z$ to $A$. The
\emph{configuration set} $\az$ is usually equipped with the
metric $d$, defined as follows
\[
d(x,y)=\frac{1}{2^n}\;\quad \text{where}\;n=\min\{i\geq 0\,:\,x_i\ne
y_i \;\text{or}\;x_{-i}\ne y_{-i}\}\enspace.
\]
The set $\az$ is a compact, totally
disconnected and perfect topological space (\ie, $\az$ is a Cantor
space).

For any pair $i,j\in\z$, with $i\leq j$, and any configuration $x\in\az$ we denote by $x_{[i,j]}$ the word
$x_i\cdots x_j\in A^{j-i+1}$, \ie, the portion of $x$ inside the interval $[i,j]$. In the previous notation, $[i,j]$ can be replaced by either $[i,j)$, or $[i,\infty)$, or $(-\infty,i]$ with the obvious meaning. A \emph{cylinder} of block $u\in A^k$ at  position $i\in\z$ is the set $[u]_i=\{x\in A^{\z}: x_{[i,i+k)}=u\}$. Cylinders are \emph{clopen} (closed and open) sets \wrt the metric $d$ and
they form a basis for the topology induced by $d$. 

The \emph{shift map}  $\sigma:\az\to\az$ is defined as $\sigma(x)_i=x_{i+1}$, for any $x\in\az$ and  $i\in\z$. The shift map is a continuos and bijective function on $\az$. The dynamical system $(\az,\sigma)$ is commonly called \emph{full shift}.
\\\\
\textbf{One-dimensional CA.}
Formally, a one dimensional \emph{Cellular Automaton} (CA) is a pair $(\az, F)$ where $F:\az\to\az$ is a continuous and $\sigma$-commuting function, \ie, $F \circ \sigma=\sigma \circ F$. Equivalently, by Hedlund's Theorem \cite{hedlund69}, a pair $(\az, F)$ is a CA if and only if there exist a natural $r\in\n$ and a map $f: A^{2r+1} \to A$ such that,
\[
\forall x\in \az,\,\forall i\in\z ,\quad F(x)_i= f(x_{i-r},
\ldots, x_{i+r})\enspace.
\]
The function $F$ is commonly called  \emph{global rule} of the CA. The natural $r$ and the map $f$ are commonly called the \emph{radius} and the \emph{local rule} of the CA, respectively.

\medskip
A CA with global rule $F$ is \emph{right} (resp., \emph{left})
\emph{closing} iff $F(x)\neq F(y)$ for any pair
$x, y\in\az$ of distinct left (resp., right) asymptotic
configurations, \ie, $x_{(-\infty,n]}=y_{(-\infty,n]}$
(resp., $x_{[n,\infty)}=y_{[n,\infty)}$) for some $n\in\z$. 
A CA is said to be
\emph{closing} if it is either left or right closing. 
Every closing CA is also surjective \cite{hedlund69}.

A rule $f:A^{2r+1} \to A$ is \emph{righmost} (resp., \emph{leftmost})
\emph{permutative} iff $\forall u\in A^{2r}, \forall\beta\in
A,\exists \alpha\in A$ such that $f(u\alpha)=\beta$ (resp.,
$f(\alpha u)=\beta$). A CA is said to be \emph{permutative} if its local rule is either rightmost or leftmost permutative.
Permutative rules are closing. 

A CA $\AZ$ is said to be \emph{right} if its local rule $f$ does not depend on the variables $x_{-r}, \ldots, x_{-1}$. In that case, $F$ can be naturally redefined as a function on $\an$ and  the pair $(\an, F)$ is a  \emph{one-sided} CA, also called the \emph{lifted version} of the two-sided CA $\AZ$~\cite{ADF09}.

The \emph{product} of two CA $(\az, F)$ and $(B^{\z}, G)$ is the CA $(\az\times B^{\z}, F\times G)$ defined as $\forall (x,y)\in \az\times B^{\z}$, $(F\times G)(x,y)=(F(x), F(y))$. The configuration space $\az\times B^{\z}$ is as usual endowed with the distance $d_{\infty}$ such that $d_{\infty}(x,y)(x',y'))=\max\{d(x,x'),d(y,y')\}$ for every pair $(x,y),(x',y')\in\az\times B^{\z}$.
\medskip

Recall that two CA $F$ and $G$ over the alphabets $A$ and $B$ are 
\emph{topologically conjugated} if
there exists a homeomorphism $\phi:\az\mapsto
B^{\z}$ such that $G\circ \phi=\phi\circ F$. The CA $F$ is a factor of the CA $G$ if  there exists a continuous and
surjective map $\phi:\az\mapsto B^{\z}$  such that $G\circ \phi=\phi\circ F$.
For any right CA $\AZ$, the one-sided lifted CA $(A^{\n}, F)$ is a factor of it.\\\\
\textbf{Additive CA.}
In this work we will focus in particular on the class of \emph{additive
CA}, \ie, CA based on an additive local rule defined over the ring
$\zm=\{0, 1, \ldots, m-1\}$. A function $f:\zm^{2r+1}\to\zm$ is
said to be additive if there exist coefficients
$a_{-r},\ldots, a_r\in\zm$ such that it can be
expressed as:
\[
\forall (x_{-r}, \ldots, x_r)\in\zm^{2r+1}, \quad f(x_{-r}, \ldots,
x_r)=\modulo{\sum_{i=-r}^r a_i x_i}{m}
\]
where $\modulo{x}{m}$ is the integer $x$ taken modulo $s$.  A CA
is \emph{additive} if its local rule is additive.
Clearly, the product of two additive CA is still additive. 

For any additive CA $(\mathbf{Z}^{\mathbb{Z}}_{m}, F)$ and any integer $p\in [2,m)$ the pair $(\mathbf{Z}^{\mathbb{Z}}_{p}, [F]_p)$ is the additive CA where $[F]_p:\mathbf{Z}^{\mathbb{Z}}_{p}\to\mathbf{Z}^{\mathbb{Z}}_{p}$ is defined as $[F]_p(x)=\modulo{F(x)}{p}$, for every $x\in\mathbf{Z}^{\mathbb{Z}}_{p}$.\\\\

\textbf{Dynamical Properties of CA.} In this subsection we review the basic background and notation on CA as dynamical systems.\\

\emph{Equicontinuous and Almost Equicontinuous CA}. Let $(\az, F)$ be a CA. A configuration $x\in\az$ is an \emph{equicontinuity point} for $F$ if $\forall\varepsilon>0$
there exists $\delta>0$ such that for all $y\in\az$, $d(x,y)<\delta$ implies that $d(F^n(y),F^n(x))<\varepsilon$ for all $n\in\n$. The existence of an equicontinuity point is related to the
existence of a special word, called \emph{blocking word}. A word $u\in A^k$ is $s$-blocking ($s\leq k$) for a CA $F$ if there exists an offset $j\in [0, k-s]$ such that for any $x,y\in [u]_0$
and any $n\in\n$, $F^n(x)_{[j,j+s)}=F^n(y)_{[j,j+s)}$\,. A word $u\in A^k$ is said to be \emph{blocking} if it is $s$-blocking for some $s\leq k$. 
$F$ is said to be \emph{equicontinuous} if $\forall\varepsilon>0$ there exists $\delta>0$ such that for all $x,y\in\az$,
$d(x,y)<\delta$ implies that $\forall n\in\n,\;d(F^n(x),F^n(y))<\varepsilon$, while
it is said to be \emph{almost equicontinuous} if the set $E$ of its equicontinuity points is residual (\ie, $E$ contains a countable intersection of dense open subsets). Recall that the CA  $F$ is equicontinuous if and only if there exist two integers $q\in\n$ and  $p>0$ such that $F^q=F^{q+p}$. If $F$ is both equicontinuous and surjective then there exists an integer $p>0$ such that $F^p(x)=x$ for all configurations $x\in\az$.\\

\emph{Sensitive to Initial Conditions CA}. Let $(\az, F)$ be a CA. The global function $F$ is \emph{sensitive to initial
conditions} (or simply \emph{sensitive}) if there exists $\varepsilon>0$ such that for any $x\in\az$ and any
$\delta>0$ there is an element $y\in X$ such that
$d(y,x)<\delta$ and $d(F^n(y),F^n(x))>\varepsilon$
for some $n\in\n$. In~\cite{Ku97}, K\r{u}rka 
proved that a one-dimensional cellular automaton is almost equicontinuous iff it is non-sensitive iff it admits a $r$-blocking word.\\

\emph{Topologically Transitive and Topologically Mixing CA}. A cellular automaton $(\az,F)$  is \emph{(topologically) mixing} if for any pair of non-empty open sets $U,V\subseteq\az$ there exists an
integer $n\in\n$ such that for any $t\geq n$ it holds that $F^t(U)~\cap~V~\ne~\emptyset$, while it is  \emph{topologically transitive} if for any
pair of non-empty open sets $U,V\subseteq\az$ there exists an
integer $n\in\n$ such that $F^n(U)\cap V\ne\emptyset$. Clearly, topological mixing implies topological transitivity. For additive cellular automata, topological transitivity is equivalent to topological mixing~\cite{CDM04}. A weaker condition than topological transitivity is the following: a CA $(\az,F)$  is \emph{non-wandering} if for any non-empty open set $U\subseteq\az$ there exists an
integer $n\in\n$ such that $F^n(U)\cap U\ne\emptyset$. In CA settings, transitivity implies surjectivity which in its turn is equivalent to the non-wandering condition~\cite{BT00}.\\

\emph{Positively Expansive CA}. A cellular automaton $(\az,F)$  is is \emph{positively
expansive} if there exists a constant $\varepsilon>0$ such that
for any pair of distinct elements  $x,y\in\az$ we have $d(F^n(x),F^n(y))\geq\varepsilon$ for some $n\in\n$.
Positively Expansive CA are left and right closing and topologically mixing, thus they
are also surjective and sensitive~\cite{Ku97,blanchard97}.  Another strong property of Positively Expansive CA
is that they are topologically conjugated to a one-sided full shift \cite{Na95}.\\\\
\textbf{Periodic orbits of CA.}
In this subsection we introduce the basic notation and basic properties for the different periodic orbits classes of CA.

Let $\AZ$ be a CA. A configuration $x\in \az$ is a \emph{temporally periodic point} of $\AZ$ if  there exists an integer $p>0$ such that $F^p(x)=x$.  
A configuration $x\in\az$ is \emph{spatially periodic} if $x$ is periodic for $\sigma$, \ie,  $\sigma^{n}(x)=x$ for some $n\in\n$. A \emph{jointly periodic} point is any configuration which is both temporally and spatially periodic. We denote by \emph{$SP(F)$}  the set of all \emph{spatially periodic configurations} of $\AZ$, with $TP(F)$ the set of all \emph{temporally periodic configurations}, with \emph{$JP(F)=SP(F)\cap TP(F)$}  the set of all \emph{jointly periodic configurations}, and with \emph{$STP(F)=TP(F) \setminus JP(F)$} the set of all \emph{strictly temporally periodic configurations of $F$}, i.e. those configurations that are temporally periodic but not spatially periodic. Note that, given a CA $\AZ$, the set $TP(F)$ is never empty. In particular, since $F(SP(F))\subseteq SP(F)$, it happens that $JP(F)$ is never empty. 

We say that a CA has \emph{dense periodic orbits} (DPO) or \emph{dense jointly periodic orbits} (JDPO) if $TP(F)$ or $JP(F)$ are dense, respectively. It is easy to obtain that surjectivity is a necessary condition for DPO or JDPO. It is still an open question whether surjectivity is a sufficient condition for DPO or JDPO \cite{BL07,Boyle08}. Among the most relevant results about this open question, we recall that all closing CA ~\cite{boyle99} and all surjective and almost equicontinuous  CA~\cite{BT00} have JDPO~\cite{boyle99}. In the case of additive CA, surjectivity implies DPO~\cite{CDM04}.

\section{Strictly temporally periodic points of surjective CA}\label{generalCA}
In this section, we consider the set of strictly temporally periodic points, $STP(F)$, for general surjective CA and we try to study its size in the different classes of increasing dynamical complexity.

In order to characterize the cardinality of the set of strictly temporally periodic points for Equicontinuous CA we need to show an easy property of
spatially periodic configurations. The following Lemma  shows that, for every CA $\AZ$, its set of spatially periodic points $SP(F)$ is \emph{meager}, i.e. negligible.

\begin{lemma} \label{residual} Let $\AZ$ be a CA. Then the set $\az \setminus SP(F)$ is residual.
\end{lemma}
\begin{proof} Recall that a residual set is the complement of a meager set and it can be equivalently defined as the countable intersection of dense open sets.
For every $w\in A^+$, let $U_w=\az \setminus \{w^\infty\}$. Clearly, every $U_w$ is an open and dense subset. 
Since $\az\setminus S=\bigcap_{w\in A^+}U_w$ and $A^+$ is countable, then the set  $\az \setminus S$ is residual.\qed
\end{proof}
\medskip

Thanks to Lemma \ref{residual} we can easily characterize the class of strictly temporally periodic orbits of Equicontinuous CA.

\begin{proposition}\label{equic} Let $(\az,F)$ be an Equicontinuous and surjective CA. Then,  $STP(F)$ is residual.
\end{proposition}
\begin{proof}
By hypothesis, there exists an integer $n>0$ such that every configuration $x\in\az$ is a temporally periodic point such that $F^n(x)=x$. Thus, it holds that $STP(F)=\az\setminus SP(F)$ and, by Lemma~\ref{residual}, it immediately follows that $STP(F)$ is residual.\qed
\end{proof}
\medskip

From Proposition \ref{equic}, the class $STP(F)$ for Equicontinuous CA is the complement of a nowhere dense set, thus it is dense in the configuration space. The property of being dense also holds for the larger class of Almost Equicontinuous CA.

\begin{proposition} \label{almost} 
Let $(\az,F)$ be an Almost Equicontinuous and surjective CA. Then, $STP(F)$ is dense.
\end{proposition}
\begin{proof}
Choose arbitrarily a configuration $x\in\az$ and an integer $k\in\n$.  Since $F$ is almost equicontinuous, there exist some integers $s,h\in\n$ such that  $F$ admits an $r$-blocking word $w\in A^s$  with offset $h$. Consider now the configuration $y=^\infty{w}uw^\infty\in [u]_0$, where $u=x_{[-k,k]}$ and without loss of generality we can assume that $u\neq w$ so that $y$ is not spatially periodic. Since surjective CA are non-wandering, there exist an integer $t>0$ and a configuration $z\in\az$ such that 
$z\in F^t([wwuww]_{-k-2s})\cap[wwuww]_{-k-2s} \neq \emptyset$. As $z$, $F^t(z)$ and $y$ belong to $[wwuww]_{-k-2s}$ and $wwuww$ is a blocking word, it follows that
\begin{align*}
F^t(y)_{[-k-2s+h, k+s+h+r)}&= F^t(z)_{[-k-2s+h, k+s+h+r)}\\
&= z_{[-k-2s+h, k+s+h+r)}\\
&= y_{[-k-2s+h, k+s+h+r)}
\end{align*}
Furthermore, also the word $ww$ is blocking. So, for any integer $i\in\z$ and any configuration $c\in[ww]_i$, it holds that 
\begin{align*}
F^t(c)_{[i+h, i+s+h+r)}&= F^t(y)_{[k+h, k+s+h+r)}\\
&= y_{[k+h, k+s+h+r)}\\
&= c_{[i+h, i+s+h+r)}
\end{align*}
Therefore, it follows that $F^{t}(y)=y$. Hence, $y$ is a temporally periodic point for $F$ such that  $d(y,x)<\frac{1}{2^k}$. Thus, the set  $STP(F)$ is dense.\qed
\end{proof}
\medskip

The following Proposition shows that there is a class of CA whose set of strictly temporally periodic orbits is empty. Recall that this property is never true for jointly periodic orbits, since every CA has at least one configuration that is both spatially and temporally periodic.

\begin{proposition} \label{exp} Let $\AZ$ be a positively expansive CA. Then $STP(F)$ is empty.
\end{proposition}
\begin{proof}
Let $\AZ$ be a positively expansive CA and let $x\in\az$ be any temporally periodic configuration for $F$. Let $t>0$ be an integer such that $F^t(x)=x$. For the sake of argument, assume now that $x\in STP(F)$, \ie, $x$ is not spatially periodic. Then, $\{\sigma^n(x)\}_{n\in\n}$ is an infinite set of distinct strictly temporally periodic points and in particular, since $F$ is $\sigma$-commuting, for each $\sigma^n(x)$ it holds that $F^t(\sigma^n(x))=\sigma^n(x)$. By the characterization of positively expansive CA from~\cite{Na95}[Thm 3.12], there exists an alphabet $B$ and a homeomorphism $\phi:\az\to B^{\n}$ such that the CA $(\az, F)$ is topologically conjugated via $\phi$ to the one-sided full shift $(B^{\n},\sigma_*)$ where $\sigma_*$ is shift map defined on $B^{\n}$. Clearly, all $\phi(\sigma^n(x))$ are distinct temporally periodic points for $\sigma_*$ such that $\sigma_*^t(\phi(\sigma^n(x)))=\phi(\sigma^n(x))$. Since there are exactly $|B|^t$ points $y\in B^{\n}$ such that $\sigma_*^t(y)=y$ but the set $\{\phi(\sigma^n(x))\}_{n\in\n}$ is infinite, we have obtained a contradiction.\qed
\end{proof}
\medskip

While the existence of an equicontinuity  point  implies a dense set of strictly temporally periodic orbits (Proposition \ref{almost}), the converse is not true, as shown by the following Proposition.

\begin{proposition} \label{sensit}
There is a sensitive CA $(\az,F)$ such that   $STP(F)$ is dense.
\end{proposition}
\begin{proof}
We show that for any surjective almost equicontinuous CA $\AZ$ and any positively expansive CA $(B^\z,G)$, the product CA $(\az \times B^\z,F\times G)$ is a sensitive to the initial conditions but non positively expansive CA such that $STP(F\times G)$ is dense. 
Let $\AZ$ and $(B^\z,G)$ be a surjective almost equicontinuous CA and a positively expansive CA, respectively. By definition, the product CA $(\az \times B^\z,F\times G)$ turns out to be a sensitive to the initial conditions CA which is not positively expansive. We now show that STP($F\times G$) is dense in $\az \times B^\z$.
For any element $(x,y) \in \az \times B^\z$ and any integer $k>0$, by Proposition~\ref{almost}, there exists a configuration $x'\in STP(F)$ such that $d(x',x)<\frac{1}{2^k}$ and, by Proposition~\ref{exp},  there exists a configuration $y' \in JP(G)$  such that $d(y',y)<\frac{1}{2^k}$. Therefore, $(x',y') \in STP(F\times G)$ and $d_{\infty}((x',x),(y',y))<\frac{1}{2^k}$. Hence, $STP(F\times G)$ is dense.\qed 
\end{proof}
\medskip

In summary, the previous result tells us that there exist classes of CA for which the set $STP(F)$ is either dense (Proposition \ref{equic} and \ref{almost}) or empty (Proposition \ref{exp}). While the existence of an equicontinuity point implies $STP(F)$ dense (Proposition \ref{almost}), the converse is generally not true (Proposition \ref{sensit}). At this point, the following two questions naturally arise.

\begin{question} Does exist a CA $\AZ$ such that $STP(F)$ is neither empty nor dense?
\end{question}

\begin{question} What is the largest class of sensitive CA where the set of strictly temporally periodic points is empty? Is this class the one of topologically transitive CA? Or, if a CA $\AZ$ is transitive, can $STP(F)$ be non empty?
\end{question}

\section{Additive CA}\label{additiveCA}
In this section, we investigate the set of strictly temporally periodic points for additive Cellular Automata. In this setting, we can provide an answer to the two questions raised in the previous section. 
In particular, we show that for additive Cellular Automata the set of strictly temporally periodic points can be either dense or empty. Moreover, we prove that it is empty if and only if the CA is topologically transitive.

\medskip
We first need to review some very useful characterizations of additive CA.  The first Theorem shows that additive CA can be decomposed in the product of \emph{simpler} additive CA, whose alphabet cardinalities are powers of prime.

\begin{theorem}[\cite{DMM03}] \label{additive} Let $(\mathbf{Z}^{\mathbb{Z}}_{pq}, F)$ be an additive CA such that  $gcd(p,q)=1$. Then, $(\mathbf{Z}^{\mathbb{Z}}_{pq}, F)$ is topologically conjugated to the additive (product) CA $(\mathbf{Z}^{\mathbb{Z}}_{p} \times \mathbf{Z}^{\mathbb{Z}}_{q}, [F]_{p} \times [F]_{q})$.
\end{theorem}
As a consequence of  the decomposition Theorem, if $m=p_1^{n_1}\cdots p_l^{n_l}$ is
the prime factor decomposition of $m$, an additive CA on $\zm$ is
topologically conjugated to the product of additive CA on $\zpini$. So all the
properties which are preserved under product and under topological
conjugacy are lifted from additive CA on $\zpk$ to $\zm$.
The following Theorem provides a strong characterization of equicontinuous and sensitive additive CA.

\begin{theorem}[\cite{MaMa99}\cite{CDF08}] \label{sensitive} Let $(\mathbf{Z}^{\mathbb{Z}}_{m}, F)$ be an additive CA with local rule $f:\zm^{2r+1}\to\zm$ defined as $f(x_{-r},...x_{r}) = [\Sigma^{r}_{i=-r}a_{i}x_{i}]_{m}$.
Then, the following statements are equivalent:
\begin{enumerate}
\item $(\mathbf{Z}^{\mathbb{Z}}_{m}, F)$ is sensitive to the initial conditions;
\item $(\mathbf{Z}^{\mathbb{Z}}_{m}, F)$ is not equicontinuous;
\item there exists a prime $p\in\n$ such that
\begin{displaymath}
p \mid m\; and \; p \nmid gcd(a_{-r}, ..., a_{-1}, a_{1}, ..., a_{r}).
\end{displaymath}
\end{enumerate}
\end{theorem}
Note that from Theorem \ref{sensitive} it immediately follows that, differently from the general case, equicontinuity/sensitivity is a dichotomy for additive CA. The following Theorem gives a characterization of surjective additive CA in terms of coefficients of the local rule.

\begin{theorem}[\cite{MaMa99}]
\label{surjective}
Let $(\mathbf{Z}^{\mathbb{Z}}_{m}, F)$ be an additive CA with local rule $f:\zm^{2r+1}\to\zm$ defined as $f(x_{-r},...x_{r}) = [\Sigma^{r}_{i=-r}a_{i}x_{i}]_{m}$.
Then, the following two statements are equivalent:
\begin{enumerate}
\item $(\mathbf{Z}^{\mathbb{Z}}_{m}, F)$ is surjective;
\item $gcd(m, a_{-r}, \ldots a_{r})=1$.
\end{enumerate}

\end{theorem}
The following Lemma expresses an other useful property for surjective additive CA and it will be used in the sequel.
\begin{lemma}[\cite{DMM03}]\label{permut} Let $(\mathbf{Z}^{\mathbb{Z}}_{p^{k}}, F)$ be a surjective additive CA with $p$ prime
and local rule  $f:\mathbf{Z}_{p^{k}}^{2r+1}\to\mathbf{Z}_{p^{k}}$ defined as $f(x_{-r},...x_{r}) = [\Sigma^{r}_{i=-r}a_{i}x_{i}]_{p^{k}}$. Set
\begin{displaymath}
L = \min\{j : gcd(a_{j}, p) = 1\}\quad and \quad R = \max\{j : gcd(a_{j}, p) = 1\}.
\end{displaymath}
Then,  there exists an integer $h \geq 1$ such that the local rule of the CA $(\mathbf{Z}^{\mathbb{Z}}_{p^{k}}, F^h)$ can be expressed by the additive map $f^h:\mathbf{Z}_{p^{k}}^{hR-hL+1}\to\mathbf{Z}_{p^{k}}$ defined as \begin{displaymath}
f^{h}(x_{hL}, ..., x_{hR}) = \left[\Sigma^{hR}_{i = hL}b_{i}x_{i} \right]_{p^{k}}
\end{displaymath}
for coefficients $b_{hL}, \ldots, b_{hR}\in\mathbf{Z}_{p^{k}}$ such that $gcd(b_{hL}, p) = gcd(b_{hR},p) = 1$.
\end{lemma}

\begin{remark}
Let $L$ and $R$ be defined in the Lemma~\ref{permut}. We want to stress that, by the surjectivity condition on the coefficients expressed by the Theorem~\ref{surjective}, both the integers $L$ and $R$ exist.
\end{remark}
We are now ready to give a classification of the strictly temporally periodic orbits for surjective additive CA whose alphabet cardinality is a power of prime.

\begin{proposition} \label{decomp} Let $(\mathbf{Z}^{\mathbb{Z}}_{p^{k}}, F)$ be a surjective  additive CA with $p$ prime and local rule  $f:\mathbb{Z}_{p^k}^{2r+1}\to\mathbb{Z}_{p^k}$ defined as $f(x_{-r},...x_{r}) = [\Sigma^{r}_{i=-r}a_{i}x_{i}]_{p^{k}}$. Then, exactly one of the following cases occurs:
\begin{itemize}
\item[$1.$] $(\mathbf{Z}^{\mathbb{Z}}_{p^{k}}, F)$ is equicontinuous and $STP(F)$ is dense,
\item[$2.$] $(\mathbf{Z}^{\mathbb{Z}}_{p^{k}}, F)$ is positively expansive and $STP(F)$ is empty,
\item[$3.$] $(\mathbf{Z}^{\mathbb{Z}}_{p^{k}}, F)$ is topologically transitive but not positively expansive and $STP(F)$ is empty.
\end{itemize}
\end{proposition}
\begin{proof}
Define $L = \min\{j : gcd(a_{j}, p) = 1\}$ and $R = \max\{j : gcd(a_{j}, p) = 1\}$. By surjectivity condition from Theorem~\ref{surjective}, $L$ and $R$ exist. By Lemma~\ref{permut}, there exists an integer $h > 0$ such that the local rule $f^h$ of the additive CA $(\mathbf{Z}^{\mathbb{Z}}_{p^{k}}, F^{h})$
can be expressed as $f^{h}(x_{hL}, ..., x_{hR}) = \left[\Sigma^{hR}_{i = hL}b_{i}x_{i} \right]_{p^{k}}$ for coefficients $b_{hL}, \ldots, b_{hR}\in\mathbf{Z}_{p^{k}}$ such that $gcd(b_{hL}, p) = gcd(b_{hR},p) = 1$. Condition $gcd(b_{hL}, p) = gcd(b_{hR},p) = 1$ implies that $f^h$ is permutative both in the leftmost variable $x_{hL}$ and the rightmost variable  $x_{hR}$. There are three possible disjoint cases:
\begin{itemize}
\item [$\mathbf{a.}$] $L=R=0$. Then, the local rule $f^{h}$ becomes $f^{h}(x_{hL}, ..., x_{hR}) = \left[\Sigma^{hR}_{i = hL}b_{i}x_{i} \right]_{p^{k}}= \left[b_0x_0\right]_{p^{k}}$, i.e.  $f^{h}$ is a permutation on $\{0,..,p^k\}$, which implies that $(\mathbf{Z}^{\mathbb{Z}}_{p^{k}}, F^h)$ is equicontinuous. Immediately follows that $(\mathbf{Z}^{\mathbb{Z}}_{p^{k}}, F)$ is equicontinuous and, by Proposition ~\ref{equic}, $STP(F)$ is dense (case 1).
\item [$\mathbf{b.}$] $L < 0 < R$. Since $f^h$ is permutative both in $x_{hL}$ and $x_{hR}$, the additive CA $(\mathbf{Z}^{\mathbb{Z}}_{p^{k}}, F^{h})$ is positively expansive which implies that also the CA 
$(\mathbf{Z}^{\mathbb{Z}}_{p^{k}}, F)$ is positively expansive and, by Proposition~\ref{exp}, $STP(F)=\emptyset$ (case 2).
\item [$\mathbf{c.}$] $0 < L \leq R$ or $L \leq R < 0$. Suppose that $0 < L \leq R$ (the case $L \leq R < 0$ is similar). So, the local rule $f^h$  of $F^h$ is one-sided and permutative in its rightmost position. Thus, the CA $(\mathbf{Z}^{\mathbb{Z}}_{p^{k}}, F^h)$ is topologically mixing~\cite{CDM02} but not positively expansive and, hence, $(\mathbf{Z}^{\mathbb{Z}}_{p^{k}}, F)$ is topologically transitive but not positively expansive. On the other hand, the lifted CA  $(\mathbf{Z}^{\n}_{p^{k}}, F^{h})$ is positively expansive. Therefore, by Proposition~\ref{exp}, it follows that $STP(F^h)=\emptyset$ for the lifted CA and then this condition holds also for $(\mathbf{Z}^{\z}_{p^{k}}, F^h)$. Since $STP(F^h)=STP(F)$, we can conclude that $(\mathbf{Z}^{\mathbb{Z}}_{p^{k}}, F)$ is topologically transitive and $STP(F)=\emptyset$ (case 3).
\end{itemize}
\qed
\end{proof}
By combining Proposition \ref{decomp} and Theorem \ref{additive} we can obtain a complete classification of the strictly temporally periodic orbits for surjective additive CA.

\begin{proposition} \label{spadd} Let $(\mathbf{Z}^{\mathbb{Z}}_{m}, F)$ be a surjective additive CA. The following statements are true:
\begin{itemize}
\item[$1.$] $STP(F)$ is either dense or empty;
\item[$2.$] $STP(F)=\emptyset$ if and only if $(\mathbf{Z}^{\mathbb{Z}}_{m}, F)$ is transitive.
\end{itemize}
\end{proposition}
\begin{proof}
Let $m=p_1^{n_1}\cdots p_l^{n_l}$ be the prime factor decomposition of $m$. By Theorem~\ref{additive}, the cellular automaton $(\mathbf{Z}^{\mathbb{Z}}_{m}, F)$  is
topologically conjugated to the product of $l$ surjective additive CA $(\zpini, [F]_{p_i^{n_i}})$. By Proposition~\ref{decomp}, each $(\zpini, [F]_{p_i^{n_i}})$ can be either equicontinuous or topologically transitive. If some cellular automaton $(\zpini, [F]_{p_i^{n_i}})$ in the decomposition is equicontinuous, then, by Proposition~\ref{equic}, it holds that $STP([F]_{p_i^{n_i}})$ is dense in $\zpini$. By topological conjugacy and since any surjective additive CA has DPO~\cite{CFMM00,CDM04}, it follows that $STP(F)$ has to be dense. Conversely, if the decomposition only contains topologically transitive CA, then, by Proposition~\ref{decomp} and the fact that topological transitivity is preserved under the product and topological conjugacy, it holds that $STP(F)$ is empty. Thus, statement 1. is true. 

Furthermore, $STP(F)$ is empty if and only if $STP([F]_{p_i^{n_i}})$ is also empty for each $(\zpini, [F]_{p_i^{n_i}})$ in the decomposition of $(\mathbf{Z}^{\mathbb{Z}}_{m}, F)$ and, by Proposition~\ref{decomp}, this happens if and only if each $(\zpini, [F]_{p_i^{n_i}})$ is topologically transitive, \ie, $(\mathbf{Z}^{\mathbb{Z}}_{m}, F)$ is topologically transitive. Therefore, statement 2. is true.\qed
\end{proof}

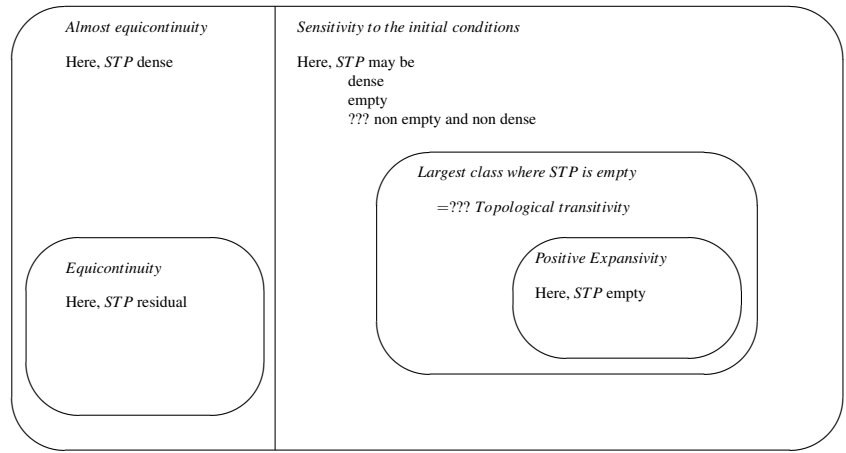
\begin{figure}
\begin{center}
\bigbreak { {\tiny    \setlength{\unitlength}{1.2pt}
\begin{picture}(282,160)
\thinlines    \put(204,58){\oval(72,38)}
              \put(138,96){$Largest\ class\ where\ STP\ is\ empty$}
	      \put(144,85){$=$??? $Topological\ transitivity$}
              \put(175,69){$Positive\ Expansivity$}
               \put(175,58){Here, $STP$ empty}
              \put(100,142){$Sensitivity\ to\ the\ initial\ conditions$}
              \put(100, 131){Here, $STP$ may be}
              \put(116,125){dense}
              \put(116,119){empty}
              \put(116,113){??? non empty and non dense}
              \put(27,66){$Equicontinuity$}
              \put(27,55){Here, $STP$ residual}
              \put(27,142){$Almost\  equicontinuity$}
              \put(27,131){Here, $STP$ dense}
              \put(52,49){\oval(75,56)}
              \put(185,69){\oval(120,70)}
              \put(93,150){\line(0,-1){140}}
              \put(141,80){\oval(262,140)}
\end{picture}}
}
\caption{The size of the set $STP(F)$ for general CA $\AZ$ belonging to the various classes of dynamical complexity. The situation simplifies in the case of additive CA.}
\label{sum}
\end{center}
\end{figure}

\section{Conclusions}\label{conclusions}
In this paper, we have studied the set of strictly temporally periodic points of surjective CA and showed that its size is inversely related to the dynamical complexity of the considered CA.  In particular, this set is residual or dense, for equicontinuous or almost equicontinuous CA, respectively, while it is empty in the class of positively expansive CA (see Figure~\ref{sum}, for a summary). Since there exist strictly sensitive to initial conditions CA with a non empty (and in particular dense) set of strictly temporally periodic points the following questions naturally arise: is there a CA such that the set of strictly temporally periodic orbits, $STP$, is neither empty nor dense? What is the largest class of sensitive  CA such that $STP$ is empty? In more general terms, can we restate the definition of topological chaos for CA in terms of strictly temporally periodic orbits? In particular, are CA with no strictly temporally periodic orbits chaotic? At the present we have no formal proof for the general case, while we can provide an answer to the above question for the class of additive CA. Indeed,  we have proved that the set of strictly temporally periodic points of additive CA is empty if and only if the cellular automaton is topologically transitive. Thus, in the additive setting, empty $STP$ implies chaotic behavior.

\end{document}

%% file: env.tex
\newtheorem{new_theorem}
{Theorem}[section]

\newtheorem{new_definition}
[new_theorem]{Definition}

\newtheorem{new_remark}
[new_theorem]{Remark}

\newtheorem{new_claim}
[new_theorem]{Claim}

\newtheorem{new_question}
[new_theorem]{Question}

\newtheorem{new_example}
[new_theorem]{Example}

\newtheorem{new_lemma}
[new_theorem]{Lemma}

\newtheorem{new_proposition}
[new_theorem]{Proposition}

\newtheorem{new_corollary}
[new_theorem]{Corollary}

\newenvironment{remark}
{\begin{new_remark}\rm}
{\end{new_remark}}

\newenvironment{question}
{\begin{new_question}\rm}
{\end{new_question}}

\newenvironment{lemma}
{\begin{new_lemma}\rm}
{\end{new_lemma}}

\newenvironment{proposition}
{\begin{new_proposition}\rm}
{\end{new_proposition}}

\newenvironment{theorem}
{\begin{new_theorem}\rm}
{\end{new_theorem}}

\newenvironment{proof}
{\medskip\noindent{\bf Proof.}$\ $}
{}

\def\squareforqed{\hbox{\rlap{$\sqcap$}$\sqcup$}}
\def\qed{\ifmmode\squareforqed\else{\unskip\nobreak\hfil
\penalty50\hskip1em\null\nobreak\hfil\squareforqed
\parfillskip=0pt\finalhyphendemerits=0\endgraf}\fi}

\newcommand{\z}{\mathbb{Z}}
\newcommand{\n}{\mathbb{N}}

\newcommand{\an}{A^{\mathbb{N}}}
\newcommand{\az}{A^{\mathbb{Z}}}
\newcommand{\AZ}{(\az,F)}